\documentclass[letterpaper,10pt,conference]{ieeeconf}
\title{\LARGE \bf Distance-coupling as an Approach to Position and Formation Control}
\author{Michael Napoli, Roberto Tron}

\usepackage[left=19mm,right=19mm,top=21.2mm,bottom=17mm]{geometry}

\usepackage[utf8]{inputenc}
\usepackage[colorlinks=TRUE,linkcolor=black,citecolor=black]{hyperref}

\usepackage[
    sorting=none,
    url=false,
    firstinits=true]{biblatex}
\usepackage{graphicx}
\graphicspath{{./figures/}}
\usepackage{subcaption}
\usepackage{dblfloatfix}

\usepackage{mathtools}
\usepackage{amsfonts, amsmath, amssymb}
\usepackage{ntheorem}
\mathtoolsset{showonlyrefs=true}
\theoremstyle{plain}
\theoremheaderfont{\bfseries}
\theoremseparator{:}

\newtheorem{prop}{Proposition}
\newtheorem{define}{Definition}
\newtheorem{corol}{Corollary}
\newtheorem{lemma}{Lemma}
\newtheorem*{remark}{Remark}

\newcommand{\mb}[1]{\mathbf{#1}}
\newcommand{\mc}[1]{\mathcal{#1}}
\newcommand{\mbb}[1]{\mathbb{#1}}
\newcommand{\R}{\mbb{R}}
\newcommand{\ls}{\Sigma}
\newcommand{\T}{\intercal}
\newcommand{\I}{\mb{I}}

\newcommand{\xeq}{x^{(\text{eq})}}
\newcommand{\xeqt}{\tilde{x}^{(\text{eq})}}
\newcommand{\Xeq}{X^{(\text{eq})}}

\newcommand{\xt}{\tilde{x}}

\newcommand{\Xt}{\tilde{X}}

\newcommand{\ut}{\tilde{u}}

\newcommand{\e}{\varepsilon}

\newcommand{\m}[1]{\langle #1 \rangle}

\begin{document}

	\maketitle
	\thispagestyle{empty}
	\pagestyle{empty}

	\begin{abstract}

        In this letter, we study the case of autonomous agents which are required to move to some new position based solely on the distance measured from predetermined reference points, or anchors. A novel approach, referred to as distance-coupling, is proposed for calculating the agent's position exclusively from differences between squared distance measurements. The key insight in our approach is that, in doing so, we cancel out the measurement's quadratic term and obtain a function of position which is linear. We apply this method to the homing problem and prove Lyapunov stability with and without anchor placement error; identifying bounds on the region of attraction when the anchors are linearly transformed from their desired positions. As an application of the method, we show how the policy can be implemented for distributed formation control on a set of autonomous agents, proving the existence of the set of equilibria.

	\end{abstract}

    \section{Introduction}
    \label{sect:introduction}

        Autonomously operated vehicles are commonly designed around high-cost measurement devices used for real-time localization. However, there exists instances, especially in low-cost systems, when the measurements available to the agents are limited to basic geometric information. We study the case where the agent's information is restricted to the Euclidean distance from $p$ fixed points, or anchors. To demonstrate, we control a single agent with linear, first-order dynamics moving from an arbitrary initial position to an equilibrium position given by target reference distances. Despite being based on nonlinear measurements, our controller is a simple linear function, and its stability can be proven in a straightforward manner. We also consider the problem of driving $m=p$ agents into some predetermined formation using only inter-agent distances. We show that the controller defined for the homing problem can be implemented with minimal adjustments. The set of equilibrium positions are then proven and validated empirically.

        \subsection{Literature Review}
        \label{sect:literature_review}

            Studies of distance-based position functions have primarily been investigated through barycentric coordinates \cite{floater_generalized_2015}. These methods represent the positions of anchors as vertices of a polygon and utilize its geometry to derive the agent's position. While this method is robust in many applications, the point is usually limited to the bounds of a convex polyhedron and maintains ill-defined behaviors when this assumption is broken. To overcome this, researchers have implemented policy switching schemes to the formation control problem with success in 2-D coordinate frames \cite{fathian_distributed_2019}.

            For the formation control problem, the most popular approach implements gradient-descent with an objective function defined on the measurement graph \cite{dimarogonas_connection_2008, dimarogonas_leaderfollower_2009}. The graph can take many forms (directed/undirected, fully connected, etc.), but the communication connections are fixed. Papers by Dimarogonas, and Choi are also among the first to demonstrate the stability of assorted gradient-based formation control policies through standard analysis methods \cite{dimarogonas_stability_2008, choi_distance-based_2021}.

            The papers \cite{ren_distributed_2007, chen_formation_2022} discuss the formation control of second-order systems, and \cite{lin_distributed_2014} defines a policy in terms of the complex Laplacian. The review paper \cite{oh_survey_2015} lays a good foundation for many of the methods available. In addition to formation control in a distance-based context, it also discusses methods for the position and displacement-based measurement cases. It discusses that, many instances the feedback laws require bearings or relative positions, so they are not purely distance-based.

        \subsection{Paper Contributions}

            In this letter, we present a new localization method which we refer to as \textit{distance-coupling}. By taking the difference between squared-distance measurements, we show that we can accurately construct a linear approximation of the agent's position in an environment given we know the positions of the reference points. Where preexisting methods depend on the convexity of the anchor environment and either bearings, or relative position information, our approach makes no such claims and allows the user to select anchor positions which are non-convex. However, we discuss degenerate collinearity cases and the minimum number of anchors required.

            For the homing problem, we prove that the linear controller exhibits global exponential convergence in the ideal case and a large region of attraction when the anchors deviate from their expected positions. For the formation control problem, we formally characterize the set of all equilibria through the controller's transformation invariance, and confirm large regions of attraction empirically. Finally, we demonstrate our results in an $\R^2$ environment;  with methods trivially generalized to any dimension.

    \section{System Definition}
    \label{sect:system_definition}

        We start by defining the agent dynamics, and the measurements available to the system at each point in time.

        \subsection{System Dynamics}
        \label{sect:system_dynamics}

            The agents are modeled as a linear, time-invariant single integrator of the form $\dot x = u,$ where $x \in \R^n$ and $u \in \R^n$ are the state and the control input, respectively. As is common in the literature, we assume that the environment is unbounded, and without obstacles. We will represent an agent's intended equilibrium position as $\xeq$.

            For the case where there are $m$ agents operating in a single environment, we can rewrite $\dot x$ in matrix form such that $\dot X = U$, where $X \in \R^{n \times m}$ and $U \in \R^{n \times m}$ are the state and control for the set of agents defined by $X = [x_1 \cdots x_m]$, and $U = [u_1 \cdots u_m]$.

        \subsection{Anchor and Anchor-distance Sets}
        \label{sect:anchor_and_anchordistance_sets}

            \begin{figure}
                \centering
                \includegraphics[width=0.35\textwidth]{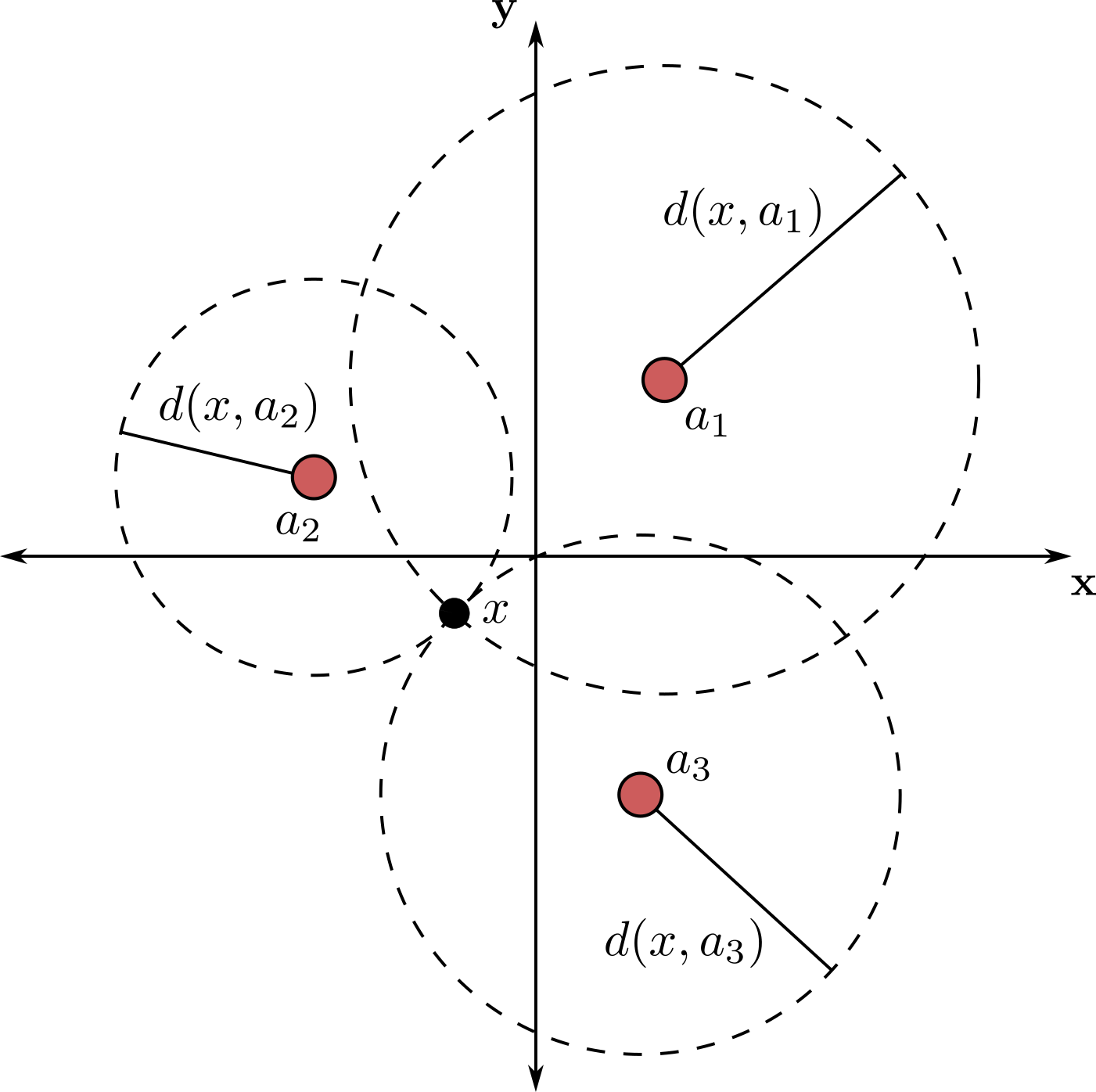}
                \caption{Anchor set with $p=3$ such that $\mc A=\{ a_1,a_2,a_3 \}$ and $\mc D(x,A)=\{ d(x,a_1), d(x,a_2), d(x,a_3) \}$.}
                \label{fig:anchor_set_example}
            \end{figure}

            We assume the presence of $p$-anchors which are represented by the set
            \begin{equation}
            \label{eq:anchor_set}
                \mc{A} = \{ a_i \in \R^n \ \forall i \leq p \},
            \end{equation}
            where $a_i$ is the position of the $i$-th anchor. We also assume that the agent measures the distance from its position to each of the anchors. The \textit{anchor-distance} set is thus defined as
            \begin{equation}
            \label{eq:anchor_distance_set}
                \mc{D}(x, \mc{A}) = \{ d(x,a_i) \in \R : a_i \in \mc{A} \},
            \end{equation}
            where $d(x,a_i) = \sqrt{ (x - a_i)^\T (x - a_i) }.$

            An example for when $p=3$ is shown in Figure \ref{fig:anchor_set_example}. In future sections, we will also refer to the anchor distances as $d_i(x) = d(x,a_i)$ and the set $\mc{D}(x) = \mc{D}(x, \mc{A})$. Again, we stress that the only measurements available to an agent is the set $\mc{D}(x)$, and their (possibly inexact) location with respect to the global environment.

    \section{Distance-coupled Position Control}
    \label{sect:distance-coupled_position_control}

        In this section we show how a linear controller can be indirectly implemented as a function of the (nonlinear) measurement set $\mc{D}(x)$ via \textit{distance-coupled} functions. We then prove Lyapunov stability, characterize the region of attraction when the anchors are linearly transformed, and show the behavior of the closed-loop system in simulation.

        \subsection{Distance-coupled Position Functions}
        \label{sect:distance-coupled_position_functions}

            We start by squaring and expanding each element of the anchor-distance set \eqref{eq:anchor_distance_set}:
            \begin{equation}
                \begin{aligned}
                    d_i^2(x) & = (x - a_i)^\T (x - a_i) \\
                    & = x^\T x - 2 x^\T a_i + a_i^\T a_i.
                \end{aligned}
            \end{equation}

            The key step in our approach is to note that by subtracting the squared-distance functions for two anchors, the nonlinear term $x^\T x$ cancels out such that
            \begin{equation}
                \begin{aligned}
                    d_i^2(x) - d_j^2(x) = &\ (x^\T x - 2 x^\T a_i + a_i^\T a_i) \\
                    & - (x^\T x - 2 x^\T a_j + a_j^\T a_j) \\
                    = & -2 x^\T (a_i - a_j) + a_i^\T a_i - a_j^\T a_j.
                \end{aligned}
            \end{equation}
            Moving the state-related terms to the left we obtain
            \begin{equation}
            \label{eq:two_anchor_differences}
                    -2  (a_i - a_j)^\T x = d_i^2(x) - d_j^2(x) - a_i^\T a_i + a_j^\T a_j.
            \end{equation}
            Note that $a_i$, and $a_j$ are expected to be known, as opposed to $d_i$ and $d_j$ which are measured.

            We reorganize the terms in \eqref{eq:two_anchor_differences} in matrix form with the intent of obtaining an expression for the full state $x$ through a linear system of equations:
            \begin{equation}
            \label{eq:policy_components}
                \begin{aligned}
                    A & = \text{vstack}( \{ -2 (a_i - a_j)^\T : \forall (i,j) \in \mc{I} \} ), \\
                    b & = \text{vstack}( \{ a_i^\T a_i - a_j^\T a_j : \forall (i,j) \in \mc{I} \} ), \\
                    h(x) & = \text{vstack}( \{ d_i^2(x) - d_j^2(x) : \forall (i,j) \in \mc{I} \} ),
                \end{aligned}
            \end{equation}
            where $\mc{I}$ is the set of $(i,j)$ indices which make up every combination of the anchors included in \eqref{eq:anchor_set}. The vstack$(\cdot)$ function concatenates a set of values into a vertically-oriented vector/matrix as appropriate. Using $|\mc{I}|$ to denote the cardinality of $\mc{I}$ we have the matrix $A \in \R^{|\mc{I}| \times n}$, and the vectors $b, h(x) \in \R^{|\mc{I}|}$.

            Combining \eqref{eq:two_anchor_differences} and \eqref{eq:policy_components}, the distance-coupled position of the agent can be computed by solving the linear system:
            \begin{equation}
            \label{eq:distance-coupled_position}
                \begin{aligned}
                    & A x = h(x) - b \Rightarrow A^\T A x = A^\T (h(x) - b) \\
                    & \Rightarrow x = K (h(x) - b),
                \end{aligned}
            \end{equation}
            where $K = (A^\T A)^{-1} A^\T$. Since the matrix $A^\T A$ is always positive definite, it is invertible if $A$ has full column rank. This condition is satisfied if there exists a subset of at least $n+1$ noncollinear anchor positions.

            We refer to \eqref{eq:distance-coupled_position} as the distance-coupled position. The corresponding distance-coupled controller is defined as
            \begin{equation}
            \label{eq:distance-coupled_control}
                u(x) = C (K (h(x) - b) - \xeq),
            \end{equation}
            where $C$ is a stable controller gain matrix, and $u(x)$ is solely dependent on the sets $\mc{A}$ (known) and $\mc{D}(x)$ (measured).

        \subsection{Global Position Stability}
        \label{sect:global_position_stability}

\begin{figure*}
    \centering
    \captionsetup{justification=centering}
    \includegraphics[width=0.90\textwidth]{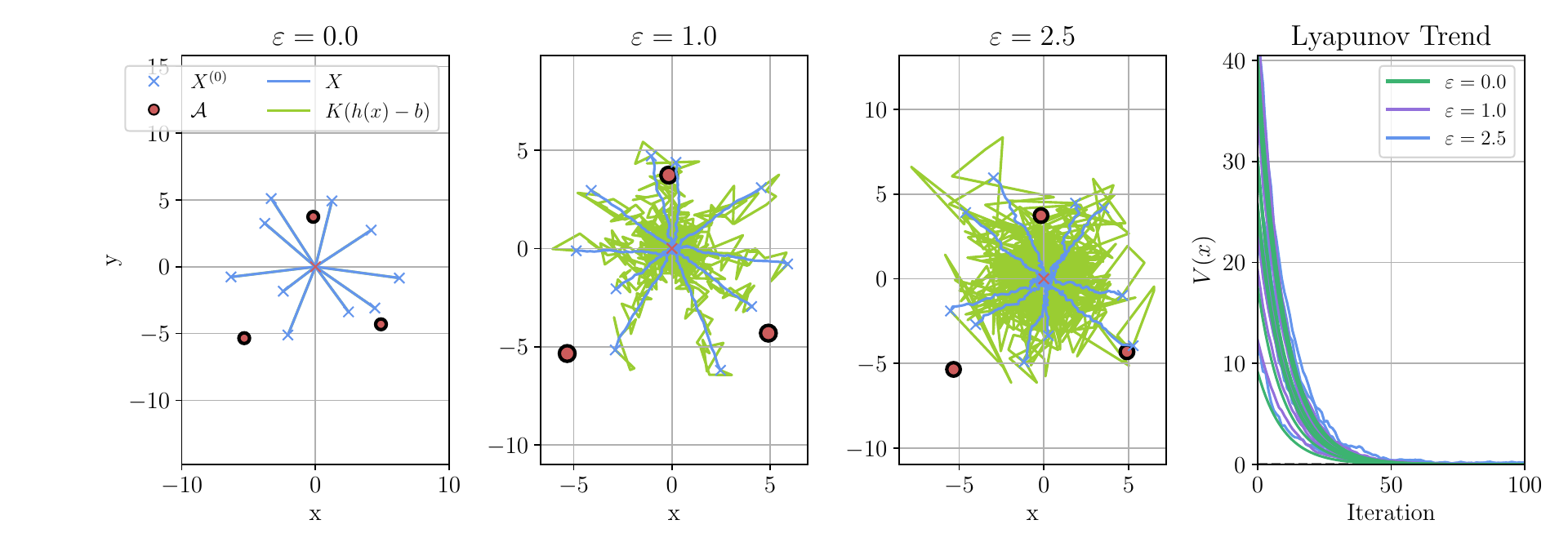}
    \caption{Response of the distance-coupled controller \eqref{eq:distance-coupled_control} with $\xeq$ at the origin and for varying magnitudes of randomly generated measurement noise, $\omega_\e$. The agents are given arbitrary initial positions around the origin.}
    \label{fig:single_homing}
\end{figure*}

            We use Lyapunov's direct method to prove stability of the distance-coupled controller \eqref{eq:distance-coupled_control} for an agent with dynamics defined by the single integrator model \cite{khalil_nonlinear_2001}.

            \begin{prop}
            \label{prop:global_anchor_stability}
                The distance-coupled control policy \eqref{eq:distance-coupled_control} is globally exponentially stable for an agent with dynamics $\dot x = u$ if $C = \alpha \I$ with $\alpha < 0$, and the identity $\I$.
            \end{prop}

            \begin{proof}
                Define $V(x) = (x - \xeq)^\T (x - \xeq)$ as the positive definite Lyapunov candidate. Then, using $\dot x = u$ and controller \eqref{eq:distance-coupled_control} with $C = \alpha \I$ it is trivial to show that $\dot V = 2 \alpha (x - \xeq)^\T (x - \xeq)$. Implying that $\dot V = 0$ when $x = \xeq$ and otherwise $\dot V < 0$ given that $\alpha < 0$. Therefore, for all $\alpha < 0$ the controller is globally exponentially stable via Lyapunov stability \cite{khalil_nonlinear_2001}.
            \end{proof}

            \begin{remark}
                Note that Proposition \ref{prop:global_anchor_stability} can be extended to the case where $C$ is any stable matrix. However, we use the restriction $C = \alpha \I$ for the sake of analysis in the next section.
            \end{remark}

        \subsection{Stability Under Transformations of the Anchor Set}
        \label{sect:stability_under_transformations_of_the_anchor_set}

            We define the agent frame as $\mbb{A}$, and the world frame $\mbb{W}$. The controller \eqref{eq:distance-coupled_control} uses anchor positions $\tilde{a}_i \in \mbb{A}$, but the measurements depend on the position $a_i \in \mbb{W}$. The two sets are related by an unknown orthogonal rotation $R$ and offset $r$ such that $a_i = R \tilde{a}_i + r$. We identify that the transformation of the set $\mc{A}$ creates error in the agent position defined by \eqref{eq:distance-coupled_position}, that is, $\xt = K(h(x) - b) = R^\T (x - r)$, where $x \in \mbb{W}$ represents the world frame position of the agent and $\xt \in \mbb{A}$ is its approximation in the anchor frame. Similarly, the controller \eqref{eq:distance-coupled_control} becomes
            \begin{equation}
            \label{eq:distance-coupled_control_shifted}
                \begin{aligned}
                    \ut(x) & = C(K(h(x) - b) - \xeq), \\
                    & = C(\xt - \xeq) = C(R^\T (x - r) - \xeq), \\
                    & = R^\T C(x - (R\xeq + r)).
                \end{aligned}
            \end{equation}

            In order to make stability claims similar to those discussed in Proposition \ref{prop:global_anchor_stability}, we will limit our system to $\R^2$, and identify the range of rotations which make $R \succ 0$.

            \begin{lemma}
            \label{lem:positive_definite_rotation}
                An orthogonal rotation $R \in \R^{2 \times 2}$ defined by $\theta \in \R$ is positive definite in the ranges $\theta \in (2 \pi k - \tfrac{\pi}{2}, 2 \pi k + \tfrac{\pi}{2})$, where $k \in \mbb{Z}$ and $\mbb{Z}$ is the set of integers.
            \end{lemma}

            \begin{proof}
                Define the orthogonal rotation $R$ in terms of Rodrigues' formula for 2-D rotations where the angle $\theta \in \R$. For some arbitrary point, $x$, we can say $R \succ 0$ when $x^\T R x > 0$ for all $x \neq 0$. After expanding, we get that $(x_1^2 + x_2^2) \cos(\theta) > 0$; letting us conclude that $R$ is positive definite so long as $\cos(\theta) > 0$. This is true for the ranges
                \begin{equation}
                \label{eq:region_of_attraction_homing}
                    \theta \in (2 \pi k - \tfrac{\pi}{2}, 2 \pi k + \tfrac{\pi}{2}) \text{ for all } k \in \mbb{Z},
                \end{equation}
                where $\mbb{Z}$ is the set of all integers.
            \end{proof}

            Using Lemma \ref{lem:positive_definite_rotation}, the stability candidate developed in Proposition \ref{prop:global_anchor_stability} can be reformulated for the rotated anchor case and used to identify the region of attraction for \eqref{eq:distance-coupled_control_shifted}.

            \begin{prop}
            \label{prop:region_of_attraction_homing}
                For a 2-D system where the set $\mc{A}$ has been transformed by an orthogonal rotation, $R$, and translation, $r$, the transformed distance-coupled controller \eqref{eq:distance-coupled_control_shifted} describes an exponentially stable policy with equilibrium $x = R \xeq + r$ for any $r$ and any $R$ such that $R^\T C \prec 0$ (where $C$ is defined in Proposition \ref{prop:global_anchor_stability} and $R$ in Lemma \ref{lem:positive_definite_rotation}).
            \end{prop}

            \begin{proof}
                Define the transformed equilibrium $\xeqt \in \R^2$ so that $\xeqt = R \xeq + r$. Placing this into our original Lyapunov candidate we have,
                \begin{equation}
                    \begin{aligned}
                        V(x) & = (x - \xeqt)^\T (x - \xeqt) \\
                        & = (x - (R \xeq + r))^\T (x - (R \xeq + r)),
                    \end{aligned}
                \end{equation}
                where $V = 0$ when $x = R \xeq + r$ and $V > 0$ otherwise. Taking the derivative of $V$ using \eqref{eq:distance-coupled_control_shifted} we get that
                \begin{equation}
                    \dot V(x) = 2 (x - (R \xeq + r))^\T R^\T C (x - (R \xeq + r)).
                \end{equation}
                By inspection, we can conclude that $\dot V(x) = 0$ when $x = R \xeq + r$, and that $\dot V \prec 0$ so long as $R^\T C \prec 0$. If we assume $C \prec 0$ by Proposition \ref{prop:global_anchor_stability}, and $R$ is selected by the bounds defined in Lemma \ref{lem:positive_definite_rotation}, then $R^\T C \prec 0$ if.
            \end{proof}

        \subsection{Distance-coupled Homing Results}
        \label{sect:distance-coupled_homing_results}

            We now validate the response of the distance-coupled controller \eqref{eq:distance-coupled_control} with varying degrees of noise in the measurement terms and different angles $\theta$ between $\mbb{A}$ and $\mbb{W}$. The claims made on the regions of attraction defined in \eqref{eq:region_of_attraction_homing} will also be demonstrated. Note that the control coefficients are defined as $C = \alpha \I$ where $\alpha = -5$.

            We define $\omega_\e \in [-\e,\e]$ as a uniformly distributed random variable. The measurement noise is thus described as $d_i(x,\e) = d_i(x) + \omega_\e$.

            Figure \ref{fig:single_homing} shows the response of the system with $p=3$ anchors and varying $\e$. In all tests we place $\xeq$ at the origin. The first case shows the ideal measurement scenario, $\e = 0$. It is clear that there is no deviation between the true position and the anchor frame approximation. Likewise, the Lyapunov function converges ($V(x) \rightarrow 0$) as predicted from Proposition \ref{prop:global_anchor_stability}. With a relatively large degree of measurement error ($\e=1$ and $\e=2.5$) we observe considerable differences between the anchor frame approximation and the true position. That said, the agent still converges to a region around the equilibrium from any point.

            Next, we validate the claims made in Section \ref{sect:stability_under_transformations_of_the_anchor_set} for transformations of the anchors without updating the policy components \eqref{eq:policy_components}. We first define the rotation $R$ with $\theta = \tfrac{\pi}{4}$, and show that the agents converge with a policy defined by $R^\T C$. We then demonstrate the effects of various translations $r$ over the anchor set and show that $x \rightarrow \xeq + r$ as the simulation progresses. Both cases are shown in Figure \ref{fig:single_offset_spin}. To demonstrate the instability defined in Proposition \ref{prop:region_of_attraction_homing} we rotate the anchors by $\theta = \tfrac{3 \pi}{4}$ which does not fall in the range defined in \eqref{eq:region_of_attraction_homing} with $k = 0$ (Figure \ref{fig:single_break}).

\begin{figure}
    \centering
    \captionsetup{justification=centering}
    \includegraphics[width=0.235\textwidth,trim={0.25cm 0.25cm 0cm 0cm},clip]{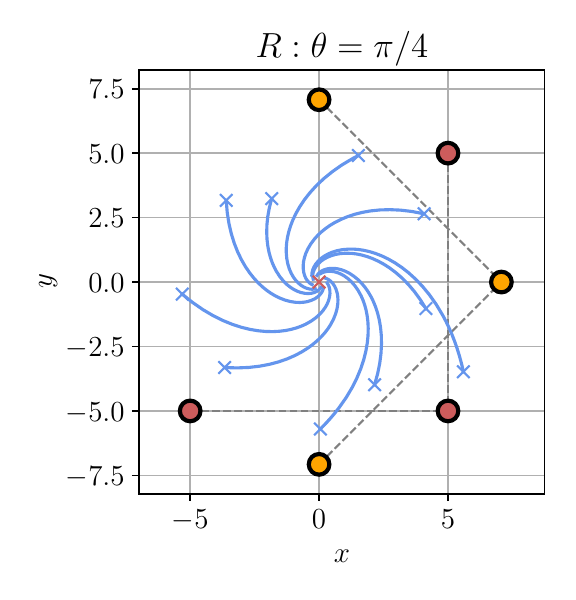}
    \includegraphics[width=0.235\textwidth,trim={0.25cm 0.25cm 0cm 0cm},clip]{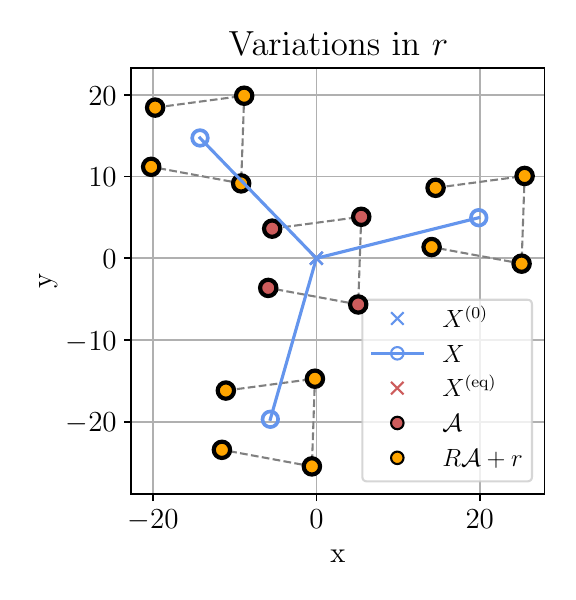}
    \caption{Homing controller response with $\xeq = 0$ and with variations in $R$ and $r$ simulated separately with (left) initial conditions around the origin, and (right) $x^{(0)} = \xeq$.}
    \label{fig:single_offset_spin}
\end{figure}

\begin{figure}
    \centering
    \captionsetup{justification=centering}
    \includegraphics[width=0.46\textwidth,trim={0.25cm 0.25cm 0cm 0cm},clip]{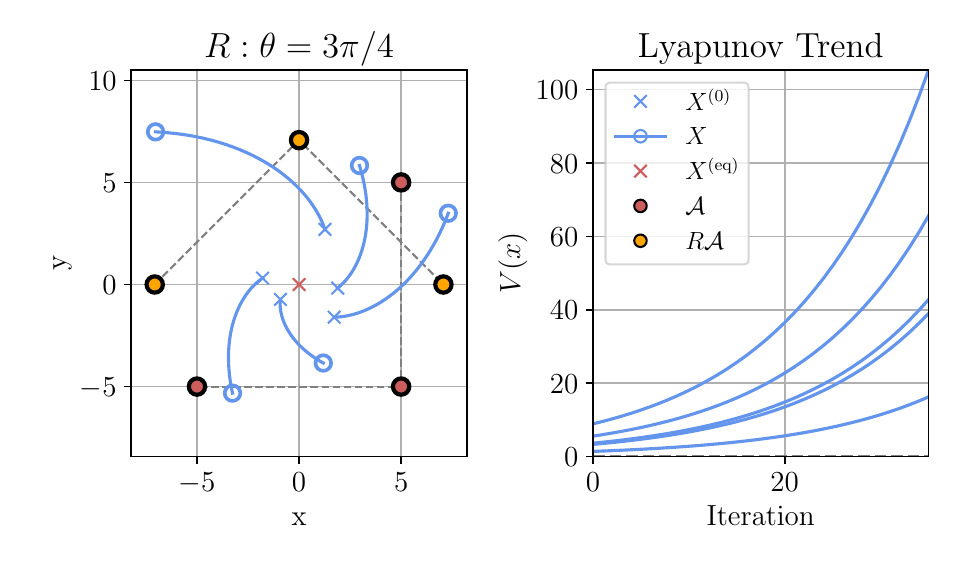}
    \caption{Distance-coupled controller response with the rotation, $R:\theta = \tfrac{3\pi}{4}$, of the anchor set (controller diverges).}
    \label{fig:single_break}
\end{figure}

    \section{Distance-coupled Formation Control}
    \label{sect:distance-coupled_formation_control}

        Using the distance-coupled homing controller \eqref{eq:distance-coupled_control}, we will define a new policy which drives $m$-agents to some predetermined formation. The only measurements each agent can make is the Euclidean distance it stands from each of the other agents in the environment. The following sections discuss the new system's structure, define the distance-coupled formation controller, and prove certain properties analytically before demonstrating in simulation.

        \subsection{Distance-coupled Formation Functions}
        \label{sect:distance-coupled_formation_functions}

            We utilize $m=p$ agents, with positions denoted by $X$ and an intended formation described by the set $\mc{A}$. By defining our formation in terms of $\mc{A}$ we can use the notation from \eqref{eq:distance-coupled_position}; approximating the positions of the agents in the anchor frame as discussed in Section \ref{sect:formation_equivalence_in_the_anchor_frame}.

            In this scenario, each agent acts as an anchor for the others and vice versa. Thus, for an agent $x_k$, we have the \textit{agent-distance} set $\mc{D}(x_k, X)$, defined similarly to \eqref{eq:anchor_distance_set}. We note that the $k$-th element of $\mc{D}(x_k, X)$ is always equal to $0$ because of its inclusion in $X$.

            We define the corresponding list of desired positions as
            \begin{equation}
                \Xeq = \begin{bmatrix} \xeq_1 & \cdots & \xeq_m \end{bmatrix} = \text{hstack}(\mc{A}),
            \end{equation}
            where hstack$(\cdot)$ is the horizontal concatenation function. In order to use \eqref{eq:distance-coupled_position} for all agents, we expand the measurement vector from \eqref{eq:policy_components} to the matrix $H(X) = [ h(x_1) \cdots h(x_m)]$,
            where $h(x_k)$ is the distance-coupled agent measurements derived from $\mc{D}(x_k,X)$, i.e.
            \begin{equation}
            \label{eq:formation_distance-coupling}
                h(x_k) = \text{vstack}( \{ d^2(x_k,x_i) - d^2(x_k,x_j) : \forall (i,j) \in \mc{I} \} ).
            \end{equation}

            We similarly arrange the squared anchor difference vector from \eqref{eq:policy_components} as $B = b \mathbf{1}_m^\T$ where $\mathbf{1}_m \in \R^m$ is a vector of all ones. This is valid as the vector $b$ is equivalent for every agent in $X$. We can then reformulate \eqref{eq:distance-coupled_position} and \eqref{eq:distance-coupled_control} as
            \begin{align}
                \label{eq:formation_position}
                & X = K (H(X) - B), \\
                \label{eq:formation_control}
                & U(X) = C(K (H(X) - B) - \Xeq),
            \end{align}
            where $K$ is unchanged from \eqref{eq:distance-coupled_position}. Finally, we can note that the only parameter which differs between the agents is the intended equilibrium position.

        \subsection{Formation Equivalence in the Anchor Frame}
        \label{sect:formation_equivalence_in_the_anchor_frame}

            Before further analysis, we define formation equivalence between the two frames of reference, $\mbb{A}$ and $\mbb{W}$, using tools from Procustes analysis. Similar to the offsets defined in Section \ref{sect:stability_under_transformations_of_the_anchor_set}, we have the agent positions in the world frame, $X \in \mbb{W}$, with corresponding anchor frame coordinates $\Xt \in \mbb{A}$.

            \begin{define}
            \label{def:formation_equivalence}
                Procrustes analysis defines the two sets, $X$ and $\Xt$, as congruent if there exists a rotation $\Psi$ such that
                \begin{equation}
                \label{eq:formation_equivalence}
                    \exists \Psi \text{ s.t. } X - \m{x} = \Psi(\Xt - \m{\xt}),
                \end{equation}
                where $\m{x} = \tfrac{1}{m} \ls_{k=1}^m x_{k}$ is the centroid of the formation.
            \end{define}

            In practice, we can compute an approximation of $\Psi$ using the Kabsch algorithm \cite{kabsch_discussion_1978}. The resulting matrix is the optimal rotation to superpose the set $\Xt$ onto $X$.

            The transformation \eqref{eq:formation_equivalence} can also be written as $\Psi X + \psi$ where $\psi = \m{x} - \Psi \m{\xt}$. In order to measure convergence, we define the $L_2$-norm of the formation error as
            \begin{equation}
            \label{eq:formation_error}
                W(X) = \ls_{k=1}^m \| (\Psi x_k + \psi - \xeq_k) \|_2.
            \end{equation}

        \subsection{Transformation Invariance}
        \label{sect:transformation_invariance}

            In order to discuss the equilibria of \eqref{eq:formation_control} we define transformation invariance and prove its application to the distance function. These claims will be used to similarly prove the invariance of \eqref{eq:formation_position} and \eqref{eq:formation_control} in the next section.

            \begin{define}
            \label{def:transformation_invariant}
                A function $f(x,y) : x,y \in \R^n$ is transformation invariant if $f(x,y)$ is constant under equivalent rotations and translations of both $x$ and $y$, that is,
                \begin{equation}
                    f(R x + r, R y + r) = f(x,y),
                \end{equation}
                for any orthogonal rotation $R \in \R^{n \times n}$ and offset $r \in \R^n$.
            \end{define}

            \begin{lemma}
            \label{lem:distance_invariance}
                The distance function,
                \begin{equation}
                \label{eq:distance}
                    d(x, y) = \sqrt{(x - y)^\T (x - y)},
                \end{equation}
                is transformation invariant by Definition \ref{def:transformation_invariant}.
            \end{lemma}

            \begin{proof}
                Follows trivially from expanding and simplifying $d(Rx + r, Ry + r)$.
            \end{proof}

        \subsection{Equilibria of the Distance-coupled Formation Controller}
        \label{sect:equilibria_of_the_distance-coupled_formation_controller}

            Using Lemma \ref{lem:distance_invariance} we establish the invariance of the distance-coupled functions in $H(X)$ before deriving the equilibria of \eqref{eq:formation_control}.

            \begin{prop}
            \label{prop:position_invariance}
                The distance-coupled position \eqref{eq:formation_position} is transformation invariant when $X = \Psi \Xt + \psi$ are congruent.
            \end{prop}

            \begin{proof}
                Identifying that the dependent variables in \eqref{eq:formation_position} are solely contained in $H(x)$, we isolate a single row of the measurement terms for the $k$-th agent, $h(x_k)$. This yields,
                \begin{equation}
                    h_{i,j}(x_k) = d^2(x_k, x_i) - d^2(x_k, x_j),
                \end{equation}
                for given indices $(i,j)$. Applying the transformation $x_k = \Psi \xt_k + \psi$ to each agent, and using Lemma \ref{lem:distance_invariance}, we get,
                \begin{equation}
                \label{eq:measurement_invariance}
                    \begin{aligned}
                        h_{i,j}(x_k) & = h_{i,j}(\Psi \xt_k + \psi) \\
                        & = d^2(\xt_k, \xt_i) - d^2(\xt_k, \xt_j) = h_{i,j}(\xt_k).
                    \end{aligned}
                \end{equation}

                Therefore, for a transformation on the entire agent group, \eqref{eq:measurement_invariance} holds for all $(i,j) \in \mc{I}$ and for any agent, $x_k$. This implies that every element in the measurement matrix, $H(X)$, is also transformation invariant. Combining this with the position given by \eqref{eq:formation_position}, we get that
                \begin{equation}
                \label{eq:position_invariance}
                    \begin{aligned}
                        & H(X) = H(\Psi \Xt + \psi) = H(\Xt) \\
                        & \begin{aligned}
                            \Rightarrow K(H(X) - B) & = K(H(\Psi \Xt + \psi) - B) \\
                            & = K(H(\Xt) - B) = \Xt.
                        \end{aligned}
                    \end{aligned}
                \end{equation}

                The claim then follows for \eqref{eq:distance-coupled_position}.
            \end{proof}

            \begin{prop}
            \label{prop:control_invariance}
                The controller \eqref{eq:formation_control} is transformation invariant when $X$ and $\Xt$ are congruent.
            \end{prop}

            \begin{proof}
                By the same logic as Proposition \ref{prop:position_invariance}, with rotation $\Psi$ and translation $\psi$, we have that
                \begin{equation}
                \label{eq:control_invariance}
                    \begin{aligned}
                        U(X) & = C(K(H(\Psi \Xt + \psi) - B) - \Xeq) \\
                        & = C(K(H(\Xt) - B) - \Xeq) = U(\Xt),
                    \end{aligned}
                \end{equation}
                implying that \eqref{eq:formation_control} is transformation invariant for any congruent $X$ and $\Xt$.
            \end{proof}

            We then use Proposition \ref{prop:control_invariance} to identify the equilibria of the distance-coupled formation controller.

            \begin{corol}
            \label{cor:formation_equilibria}
                The formation controller \eqref{eq:formation_control} is at an equilibrium, $U(X) = 0$, when $X$ is congruent to $\Xeq$.
            \end{corol}

            \begin{proof}
                Define $\Xt$ at $\Xeq$ such that $U(\Xt) = 0$. Then, by Proposition \ref{prop:control_invariance} we have that $U(X) = U(\Psi \Xt + \psi) = U(\Xt) = 0$. Thus, the formation $X = \Psi \Xt + \psi$ is at an equilibrium given it is congruent to $\Xeq$.
            \end{proof}

            It should be noted that, while Corollary \ref{cor:formation_equilibria} defines the equilibria of \eqref{eq:formation_control}, it makes no claims on its stability. However, because of the similarities between \eqref{eq:distance-coupled_control} and \eqref{eq:formation_control}, it can be reasonably concluded that, for agent positions $X$ which are not congruent to $\Xt$, we get similar bounds on the region of attraction as those defined in Lemma \ref{lem:positive_definite_rotation}. This claim is validated empirically in Section \ref{sect:distance-coupled_formation_results}.

        \subsection{Distance-coupled Formation Results}
        \label{sect:distance-coupled_formation_results}

\begin{figure}[b]
    \centering
    \captionsetup{justification=centering}
    \includegraphics[width=0.48\textwidth]{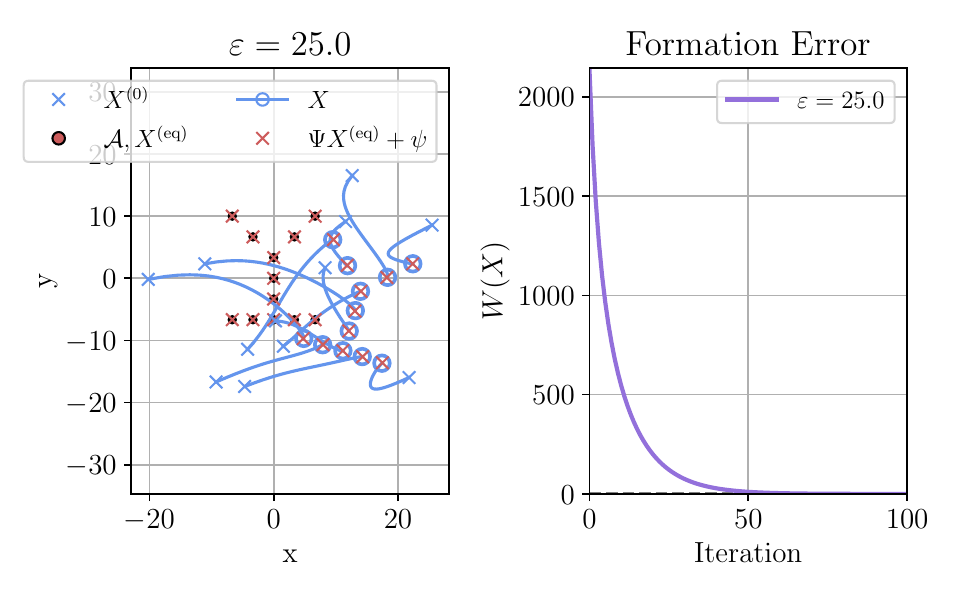}
    \caption{Distance-coupled formation control response for an arbitrary formation and with varying levels of randomly generated offsets from the equilibrium positions: $X^{(0)} = \Xeq + \omega_{\e = 25}$.}
    \label{fig:multi_formation}
\end{figure}

\begin{figure}
    \centering
    \captionsetup{justification=centering}
    \includegraphics[width=0.48\textwidth,trim={0.25cm 0.25cm 0cm 0cm},clip]{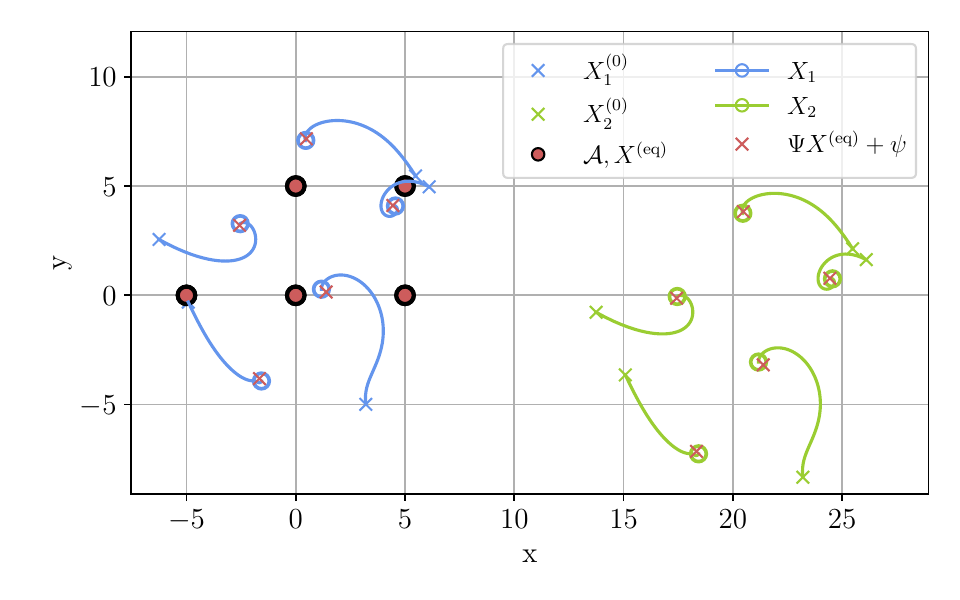}
    \caption{Response of formation controller with $X_1 = \Xeq + \omega_{\e = 10}$ and $X_2 = X_1 + r$.}
    \label{fig:multi_offset}
\end{figure}

\begin{figure}
    \centering
    \captionsetup{justification=centering}
    \includegraphics[width=0.48\textwidth,trim={0.25cm 0.25cm 0cm 0cm},clip]{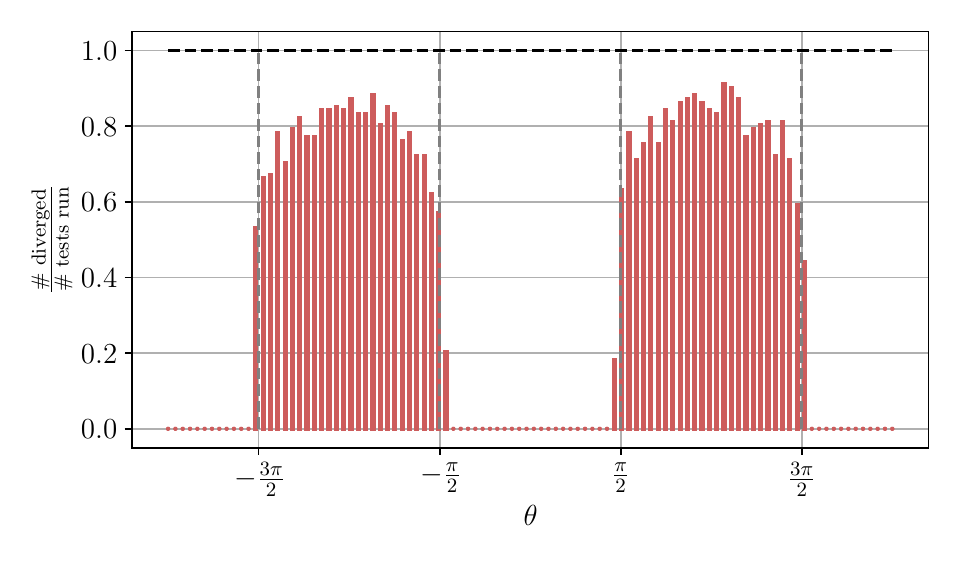}
    \caption{Region of attraction results with $p=4$ and $100$ repetitions at each $R$ with $\theta \in [-2\pi,2\pi]$ and for an initial position $X^{(0)} = \Xeq + \omega_{\e = 0.1}$.}
    \label{fig:multi_rotation}
\end{figure}

            While we do not present a proof for the stability of the distance-coupled formation controller, we can demonstrate empirically its behavior under various initial conditions. We first show the response of the policy when the initial conditions are randomly offset around $\Xeq$. This will also help to validate our claims on the equilibrium positions defined in Corollary \ref{cor:formation_equilibria}.

            To validate the response of \eqref{eq:formation_control} we initialize the agent group such that
            $$
                X^{(0)} = R(\Xeq + \omega_\e),
            $$
            with noise value $\e = 25$, and rotation $R = \I$. The simulation is then iterated and the behavior of the policy is shown. The performance of the formation error function $W(X)$ is plotted over time; showing that it converges in this case.

            Figure \ref{fig:multi_offset} uses two sets of agents to validate the convergence of the formation controller for similarly perturbed formations. The first set, $X_1$, is the reference set and has initial conditions which are generated similar to the Figure \ref{fig:multi_formation} example with $\e=10$. The second set $X_2$ is linearly offset from $X_1$ by a constant $r$. As expected, the two sets follow the same trajectory and end in equilibrium positions which differ by the initial translation, $r$.

            Because of their close relationship, we tested if the controller would behave similarly to rotations about the origin as it does in the homing problem. We thus take $\e = 0.1$ and select $R$ for a uniformly distributed set $\theta \in [-2\pi, 2\pi]$. The simulation was then iterated until the controller converged ($W(X) \rightarrow 0$) or diverged ($W(X) \rightarrow \infty$), and the proportion of tests which diverged for $100$ samples are shown in Figure \ref{fig:multi_rotation}. As conjectured, there is a strong correlation to the bounds defined in \eqref{eq:region_of_attraction_homing}.

    \section{Conclusion}

        In this paper we define the distance-coupling method as a suitable approach to the position and formation control problems. In the homing case, we prove stability and show that it is held for linear transformations of the anchor positions, even when the policy components \eqref{eq:policy_components} are not updated. The distance-coupled formation controller is also studied, and the equilibrium conditions are established under the definitions of congruence and transformation invariance. Claims are made on the general performance of the system under arbitrary transformations, and the resulting convergence properties are demonstrated empirically. A key feature in our approach is that the controller can be implemented without relative position or bearings.

        Generally speaking, the approach described here is framed using a single integrator model for simplicity. That said, the distance-coupled position functions \eqref{eq:distance-coupled_position} and \eqref{eq:formation_position} can be applied to any localization problem in which the available measurements are a function of the distance from known reference points. In the context of systems with more complex dynamics, the distance-coupled policies defined here can be used to calculate suitable waypoints for the agents to follow by a more appropriate controller.

    \begin{raggedright}
        \printbibliography

@inproceedings{dimarogonas_stability_2008,
	title = {On the stability of distance-based formation control},
	doi = {10.1109/CDC.2008.4739215},
	booktitle = {2008 47th {IEEE} {Conference} on {Decision} and {Control}},
	author = {Dimarogonas, Dimos V. and Johansson, Karl H.},
	month = dec,
	year = {2008},
	note = {ISSN: 0191-2216},
	pages = {1200--1205},
}

@article{ren_distributed_2007,
	title = {Distributed multi-vehicle coordinated control via local information exchange},
	volume = {17},
	copyright = {Copyright © 2006 John Wiley \& Sons, Ltd.},
	issn = {1099-1239},
	url = {https://onlinelibrary.wiley.com/doi/abs/10.1002/rnc.1147},
	doi = {10.1002/rnc.1147},
	language = {en},
	number = {10-11},
	urldate = {2023-08-18},
	journal = {International Journal of Robust and Nonlinear Control},
	author = {Ren, Wei and Atkins, Ella},
	year = {2007},
	note = {\_eprint: https://onlinelibrary.wiley.com/doi/pdf/10.1002/rnc.1147},
	pages = {1002--1033},
}

@inproceedings{chen_formation_2022,
	title = {Formation control for second-order nonlinear multi-agent systems with external disturbances via adaptive method},
	url = {https://ieeexplore.ieee.org/document/10055393},
	doi = {10.1109/CAC57257.2022.10055393},
	urldate = {2023-10-08},
	booktitle = {2022 {China} {Automation} {Congress} ({CAC})},
	author = {Chen, Ting and Sun, Yujie and Niu, Xinglong and Lan, Yanting and Fang, Wei and Liu, Peng},
	month = nov,
	year = {2022},
	note = {ISSN: 2688-0938},
	pages = {5616--5620},
}

@article{oh_survey_2015,
	title = {A survey of multi-agent formation control},
	volume = {53},
	issn = {0005-1098},
	url = {https://www.sciencedirect.com/science/article/pii/S0005109814004038},
	doi = {10.1016/j.automatica.2014.10.022},
	urldate = {2023-10-13},
	journal = {Automatica},
	author = {Oh, Kwang-Kyo and Park, Myoung-Chul and Ahn, Hyo-Sung},
	month = mar,
	year = {2015},
	pages = {424--440},
}

@article{lin_distributed_2014,
	title = {Distributed {Formation} {Control} of {Multi}-{Agent} {Systems} {Using} {Complex} {Laplacian}},
	volume = {59},
	issn = {0018-9286, 1558-2523},
	url = {http://ieeexplore.ieee.org/document/6750042/},
	doi = {10.1109/TAC.2014.2309031},
	language = {en},
	number = {7},
	urldate = {2023-10-13},
	journal = {IEEE Transactions on Automatic Control},
	author = {Lin, Zhiyun and Wang, Lili and Han, Zhimin and Fu, Minyue},
	month = jul,
	year = {2014},
	pages = {1765--1777},
}

@book{khalil_nonlinear_2001,
	edition = {3rd},
	title = {Nonlinear {Systems}},
	isbn = {978-0-13-067389-3},
	url = {https://www.pearson.com/en-us/subject-catalog/p/nonlinear-systems/P200000003306/9780130673893},
	language = {English},
	urldate = {2023-10-14},
	publisher = {Pearson},
	author = {Khalil, Hassan},
	month = dec,
	year = {2001},
}

@article{dimarogonas_leaderfollower_2009,
	title = {Leader–follower cooperative attitude control of multiple rigid bodies},
	volume = {58},
	issn = {0167-6911},
	url = {https://www.sciencedirect.com/science/article/pii/S0167691109000255},
	doi = {10.1016/j.sysconle.2009.02.002},
	number = {6},
	urldate = {2023-10-18},
	journal = {Systems \& Control Letters},
	author = {Dimarogonas, Dimos V. and Tsiotras, Panagiotis and Kyriakopoulos, Kostas J.},
	month = jun,
	year = {2009},
	pages = {429--435},
}

@article{dimarogonas_connection_2008,
	title = {A connection between formation infeasibility and velocity alignment in kinematic multi-agent systems},
	volume = {44},
	issn = {0005-1098},
	url = {https://www.sciencedirect.com/science/article/pii/S0005109808002343},
	doi = {10.1016/j.automatica.2008.03.013},
	number = {10},
	urldate = {2023-10-26},
	journal = {Automatica},
	author = {Dimarogonas, Dimos V. and Kyriakopoulos, Kostas J.},
	month = oct,
	year = {2008},
	pages = {2648--2654},
}

@inproceedings{fathian_distributed_2019,
	title = {Distributed {Formation} {Control} via {Mixed} {Barycentric} {Coordinate} and {Distance}-{Based} {Approach}},
	url = {https://ieeexplore.ieee.org/abstract/document/8814890},
	doi = {10.23919/ACC.2019.8814890},
	urldate = {2023-10-26},
	booktitle = {2019 {American} {Control} {Conference} ({ACC})},
	author = {Fathian, Kaveh and Rachinskii, Dmitrii I. and Spong, Mark W. and Summers, Tyler H. and Gans, Nicholas R.},
	month = jul,
	year = {2019},
	note = {ISSN: 2378-5861},
	pages = {51--58},
}

@article{floater_generalized_2015,
	title = {Generalized barycentric coordinates and applications},
	volume = {24},
	issn = {0962-4929, 1474-0508},
	url = {https://www.cambridge.org/core/journals/acta-numerica/article/generalized-barycentric-coordinates-and-applications/62321A5281E90552020B728B0F5D4D0C},
	doi = {10.1017/S0962492914000129},
	language = {en},
	urldate = {2023-10-25},
	journal = {Acta Numerica},
	author = {Floater, Michael S.},
	month = may,
	year = {2015},
	note = {Publisher: Cambridge University Press},
	pages = {161--214},
}

@article{choi_distance-based_2021,
	title = {Distance-{Based} {Formation} {Control} {With} {Goal} {Assignment} for {Global} {Asymptotic} {Stability} of {Multi}-{Robot} {Systems}},
	volume = {6},
	issn = {2377-3766},
	url = {https://ieeexplore.ieee.org/document/9360428},
	doi = {10.1109/LRA.2021.3061071},
	number = {2},
	urldate = {2023-10-30},
	journal = {IEEE Robotics and Automation Letters},
	author = {Choi, Yun Ho and Kim, Doik},
	month = apr,
	year = {2021},
	note = {Conference Name: IEEE Robotics and Automation Letters},
	pages = {2020--2027},
}

@article{kabsch_discussion_1978,
	title = {A discussion of the solution for the best rotation to relate two sets of vectors},
	volume = {34},
	copyright = {Copyright (c) 1978 International Union of Crystallography},
	issn = {0567-7394},
	url = {//scripts.iucr.org/cgi-bin/paper?a15629},
	doi = {10.1107/S0567739478001680},
	language = {en},
	number = {5},
	urldate = {2023-11-12},
	journal = {Acta Crystallographica Section A: Crystal Physics, Diffraction, Theoretical and General Crystallography},
	author = {Kabsch, W.},
	month = sep,
	year = {1978},
	note = {Number: 5
Publisher: International Union of Crystallography},
	pages = {827--828},
}
    \end{raggedright}

\end{document}